\documentclass{article}
\usepackage{arxiv}
\usepackage{bbm,fullpage}
\usepackage[T1]{fontenc}
\usepackage[utf8]{inputenc}
\usepackage[english]{babel}
\usepackage{amsmath,amssymb,amsfonts,amsthm}
\usepackage{cmap}
\usepackage{microtype}
\usepackage{mathtools}
\usepackage{mathrsfs}
\usepackage{bm}

\usepackage{url}            
\usepackage{booktabs}       
\usepackage{nicefrac}       
\usepackage{microtype}      

\usepackage[inline,shortlabels]{enumitem}

\usepackage[linesnumbered,ruled,vlined]{algorithm2e}

\SetCommentSty{mycommfont}
\SetKwInput{KwData}{Input}
\SetKwInput{KwResult}{Output}
\SetKwInput{KwAssumption}{Assumption}

\newtheorem{proposition}{Proposition}
\newtheorem{lemma}[proposition]{Lemma}

\newtheorem{theorem:intro}{Theorem}

\newcommand{\KK}{{\bm{K}}}
\newcommand{\CC}{{\bm{C}}}
\newcommand{\QQ}{{\bm{Q}}}
\newcommand{\RR}{{\bm{R}}}

\newcommand{\bmA}{{\bm{A}}}

\newcommand{\sZ}{{\mathscr{Z}}}

\newcommand{\bN}{{\mathbb{N}}}

\newcommand{\bZ}{{\mathbb{Z}}}
\newcommand{\cI}{{\mathcal{I}}}
\newcommand{\cL}{{\mathcal{L}}}

\newcommand{\cC}{{\mathcal{C}}}

\newcommand{\cH}{{\mathcal{H}}}
\newcommand{\cK}{{\mathcal{K}}}
\newcommand{\cR}{{\mathcal{R}}}
\newcommand{\cS}{{\mathcal{S}}}

\newcommand{\cP}{{\mathcal{P}}}

\newcommand{\cV}{{\mathcal{V}}}
\newcommand{\cW}{{\mathcal{W}}}

\newcommand{\BB}{{B}}

\newcommand{\VV}{{V}}
\newcommand{\ZZ}{{Z}}

\newcommand{\kK}{{\mathfrak{K}}}
\newcommand{\kG}{{\mathfrak{G}}}
\newcommand{\kP}{{\mathfrak{P}}}
\newcommand{\kQ}{{\mathfrak{Q}}}
\newcommand{\kR}{{\mathfrak{R}}}
\newcommand{\kS}{{\mathfrak{S}}}

\newcommand{\nf}{{g}}

\newcommand{\x}{{\bm{x}}}
\newcommand{\y}{{\bm{y}}}

\newcommand{\eps}{{\varepsilon}}

\newcommand{\ext}{{\mathrm{ext}}}

\newcommand{\Reps}{\RR\langle \eps\rangle}
\newcommand{\Ceps}{\CC\langle \eps\rangle}

\newcommand{\diff}[2]{\frac{\partial {#1}}{\partial {#2}}}

\allowdisplaybreaks[4]

\title{Computing the Real Isolated Points of an \\ Algebraic Hypersurface}
\date{\today}

\author{Huu Phuoc {Le} \\
 Sorbonne Universit\'e, \textsc{CNRS},\\
    Laboratoire d'Informatique de Paris~6, \textsc{LIP6}, \\
    \'Equipe \textsc{PolSys} \\
    F-75252, Paris Cedex 05, France \\
    \texttt{huu-phuoc.le@lip6.fr}
\And
Mohab {Safey El Din} \\
Sorbonne Universit\'e, \textsc{CNRS},\\
Laboratoire d'Informatique de Paris~6, \textsc{LIP6}, \\
\'Equipe \textsc{PolSys} \\
F-75252, Paris Cedex 05, France \\
\texttt{mohab.safey@lip6.fr}
\And
Timo de Wolff \\
Technische Universit\"at Braunschweig,\\
Institut f\"ur Analysis und Algebra, AG Algebra \\
Universit\"splats 2 \\
38106 Braunschweig, Germany \\
\texttt{t.de-wolff@tu-braunschweig.de}
}

\begin{document}
\maketitle

\begin{abstract}
  Let $\RR$ be the field of real numbers. We consider the problem of
computing the real isolated points of a real algebraic set in $\RR^n$
given as the vanishing set of a polynomial system. This problem plays
an important role for studying rigidity properties of mechanism in
material designs. In this paper, we design an algorithm which solves
this problem. It is based on the computations of critical points as
well as roadmaps for answering connectivity queries in real algebraic
sets. This leads to a probabilistic algorithm of complexity
$(nd)^{O(n\log(n))}$ for computing the real isolated points of real
algebraic hypersurfaces of degree $d$. It allows
us to solve in practice instances which are out of reach of the
state-of-the-art.
\end{abstract}


\footnotesize{\thanks{Mohab Safey El Din and Huu Phuoc Le are
supported by the ANR grants ANR-18-CE33-0011 \textsc{Sesame}, and
ANR-19-CE40-0018 \textsc{De Rerum Natura}, the joint ANR-FWF
ANR-19-CE48-0015 \textsc{ECARP} project, the PGMO grant
\textsc{CAMiSAdo} and the European Union’s Horizon 2020 research and
innovative training network programme under the Marie
Sk\l{}odowska-Curie grant agreement N${}^{\circ}$ 813211 (POEMA). Timo
de Wolff is supported by the DFG grant WO 2206/1-1.}}

\section{Introduction}

Let $\QQ$, $\RR$ and $\CC$ be respectively the fields of rational, real and
complex numbers. For $\x\in \RR^n$ and $r \in \RR$, we denote by $B(\x,
r)\subset \RR^n$ the open ball centered at $\x$ of radius $r$.

Let $f \in \QQ[x_1, \ldots, x_n]$ and $\cH\subset \CC^n$ be the hypersurface
defined by $f = 0$. We aim at computing the {\em isolated points} of
$\cH\cap \RR^n$, i.e. the set of points $\x\in \cH\cap \RR^n$ s.t. for some
positive $r$, $B(\x, r)\cap \cH = \{\x\}$. We shall denote this set
of isolated real points by $\sZ(\cH)$.

\paragraph*{Motivation}
We consider here a particular instance of the more general problem of
computing the isolated points of a {\em semi-algebraic} set. Such
problems arise naturally and frequently in the design of rigid
mechanism in material design. Those are modeled canonically with
semi-algebraic constraints, and isolated points to the semi-algebraic
set under consideration are related to mobility/rigidity properties of
the mechanism.  A particular example is the study of \textit{auxetic}
materials, i.e., materials that shrink in all directions under
compression. These materials appear in nature (first discovered in
\cite{La87}) e.g., in foams, bones or propylene; see
e.g. \cite{YangEtAl:ReviewAuxeticMaterials}, and have various
potential applications. They are an active field of research, not only
on the practical side, e.g., \cite{Gr06, GasparEtAl:Honeycombs}, but
also with respect to mathematical foundations; see e.g.
\cite{Borcea:Streinu:GeometricAuxetics,
Borcea:Streinu:PeriodicAuxetics}. On the constructive side, these
materials are closely related to \textit{tensegrity frameworks}, e.g.,
\cite{Roth:Whiteley:Tensegrity, CoWh92}, which can possess various
sorts of rigidity properties.

Hence, we aim to provide a practical algorithm for computing these
real isolated points in the particular case of real traces of complex
hypersurfaces first. This simplification allows us to significantly
improve the state-of-the-art complexity for this problem and to
establish a new algorithmic framework for such computations.
\vspace{-0.3cm}
\paragraph*{State-of-the-art} As far as we know, there is no established
algorithm dedicated to the problem under consideration here. However, effective
real algebraic geometry provides subroutines from which such
a computation could be done. Let $\cH$ be a hypersurface defined by $f = 0$ with
$f \in \QQ[x_1, \ldots, x_n]$ of degree $d$.

A first approach would be to compute a cylindrical algebraic decomposition
adapted to $\cH\cap \RR^n$~\cite{Collins76}. It partitions $\cH\cap\RR^d$ into
connected {\em cells}, i.e. subsets which are homeomorphic to $]0,1[^i$ for some
$1\leq i\leq n$. Next, one needs to identify cells which correspond to isolated
points using adjacency information~(see e.g. \cite{Arnon}). Such a procedure is
at least doubly exponential in $n$ and polynomial in $d$.

A better alternative is to encode real isolated points with quantified formula
over the reals. Using e.g. \cite[Algorithm 14.21]{BPR}, one can compute isolated
points of $\cH\cap \RR^n$ in time $d^{O(n^2)}$. Note also that \cite{Vor99}
allows to compute isolated points in time $d^{O(n^3)}$.

A third alternative (suggested by the reviewers) is to use
\cite[Algorithm 12.16]{BPR} to compute sample points in each connected
component of $\cH\cap \RR^n$ and then decide whether spheres, centered
at these points, of infinitesimal radius, meet $\cH\cap \RR^n$. Note
that these points are encoded with parametrizations of degree
$d^{O(n)}$ (their coordinates are evaluations of polynomials at the
roots of a univariate polynomial with infinitesimal
coefficients). Applying \cite[Alg. 12.16]{BPR} on this last real root
decision problem would lead to a complexity $d^{O(n^2)}$ since the
input polynomials would have degree $d^{O(n)}$. Another approach would
be to run \cite[Alg. 12.16]{BPR} modulo the algebraic extension used
to define the sample points. That would lead to a complexity
$d^{O(n)}$ but this research direction requires modifications of
\cite[Alg. 12.16]{BPR} since it assumes the input coefficients to lie
in an {\em integral domain}, which is not satisfied in our
case. Besides, we report on practical experiments showing that using
\cite[Alg. 12.16]{BPR} to compute {\em only} sample points in $\cH\cap
\RR^n$ does not allow us to solve instances of moderate size.

The topological nature of our problem is related to connectedness. Computing
isolated points of $\cH\cap \RR^n$ is equivalent to computing those connected
components of $\cH\cap \RR^n$ which are reduced to a single point (see Lemma
\ref{lemma:singleton}). Hence, one considers computing {\em roadmaps}: these are
algebraic curves contained in $\cH$ which have a non-empty and connected
intersection with all connected components of the real set under study. Once
such a roadmap is computed, it suffices to compute the isolated points of a
semi-algebraic curve in $\RR^n$. This latter step is not trivial; as many of the
algorithms computing roadmaps output either curve segments (see e.g.,
\cite{BRSS14}) or algebraic curves (see e.g., \cite{SaSc17}). Such curves are
encoded through {\em rational parametrizations}, i.e., as the Zariski closure of
the projection of the $(x_1, \ldots, x_n)$-space of the solution set to
\[\textstyle{w(t, s) = 0, x_i = v_i(t, s) / \frac{\partial w}{\partial
t} (t, s), \quad 1\leq i \leq n}\]
where $w\in \QQ[t, s]$ is square-free and monic in $t$ and $s$ and the
$v_i$'s lie in $\QQ[t, s]$ (see e.g., \cite{SaSc17}).  As far as we
know, there is no published algorithm for computing isolated points
from such an encoding.

Computing roadmaps started with Canny's (probabilistic) algorithm
running in time $d^{O(n^2)}$ on real algebraic sets. Later on,
\cite{SaSc11} introduced new types of connectivity results enabling
more freedom in the design of roadmap algorithms. This led to
\cite{SaSc11, BRSS14} for computing roadmaps in time
$(nd)^{O(n^{1.5})}$. More recently, \cite{BR14}, still using these new
types of connectivity results, provide a roadmap algorithm running in
time $d^{O(n\log^2 n)}n^{O(n\log^3 n)}$ for general real algebraic
sets (at the cost of introducing a number of infinitesimals). This is
improved in \cite{SaSc17}, for smooth bounded real algebraic sets,
with a probabilistic algorithm running in time $O((nd)^{12n\log_2
n})$. These results makes plausible to obtain a full algorithm running
in time $(nd)^{O(n\log n)}$ to compute the isolated points of
$\cH\cap\RR^n$.
\paragraph*{Main result} We provide a probabilistic algorithm which takes as
input $f$ and computes the set of real isolated points $\sZ(\cH)$ of $\cH\cap
\RR^n$. A few remarks on the output data-structure are in order.
\noindent
Any finite algebraic set $Z\subset\CC^n$ defined over $\QQ$ can be represented
as the projection on the $(x_1, \ldots, x_n)$-space of the solution set to
\[\textstyle{w(t) = 0, x_i = v_i(t), \quad 1\leq i \leq n}\]
where $w\in \QQ[t]$ is square-free and the $v_i$'s lie in
$\QQ[t]$. The sequence of polynomials $(w, v_1, \ldots, v_n)$ is
called a {\em zero-dimensional parametrization}; such a representation
goes back to Kronecker \cite{Kr82}. Such representations (and their
variants with denominators) are widely used in computer algebra (see
e.g. \cite{GianniM87, GiustiHMP95, GLS01}). For a zero-dimensional
parametrization $\kQ$, $\ZZ(\kQ)\subset \CC^n$ denotes the finite set
represented by $\kQ$. Observe that considering additionally isolating
boxes, one can encode $\ZZ(\kQ)\cap \RR^n$. Our main result is as
follows.
\begin{theorem:intro}\label{thm:main}
  Let $f\in \QQ[x_1, \ldots, x_n]$ of degree $d$ and $\cH\subset
\CC^n$ be the algebraic set defined by $f=0$. There exists a
probabilistic algorithm which, on an input $f$ of degree $d$, computes
a zero-dimensional parametrization $\kP$ and isolating boxes which
encode $\sZ(\cH)$ using $(nd)^{O(n\log(n))}$ arithmetic operations in
$\QQ$.
\end{theorem:intro}
In Section~\ref{sec:experiments}, we report on practical experiments
showing that it already allows us to solve non-trivial problems which
are actually out of reach of \cite[Alg. 12.16]{BPR} to compute sample
points in $\cH\cap\RR^n$ only. We sketch now the geometric ingredients
which allow us to obtain such an algorithm. Assume that $f$ is
non-negative over $\RR^n$ (if this is not the case, just replace it by
its square) and let $\x\in \sZ(\cH)$. Since $\x$ is isolated and $f$
is non-negative over $\RR^n$, the intuition is that for $e>0$ and
small enough, the real solution set to $f=e$ looks like a ball around
$\x$, hence a bounded and closed connected component $C_\x$. Then the
restriction of every projection on the $x_i$-axis to the algebraic set
$\cH_e\subset \CC^n$ defined by $f=e$ intersects $C_\x$. When $e$
tends to $0$, these critical points in $C_\x$ ``tend to $\x$''. This
first process allows us to compute a superset of candidate points in
$\cH\cap \RR^n$ containing $\sZ(\cH)$. Of course, one would like that
this superset is finite and this will be the case up to some generic
linear change of coordinates, using e.g. \cite{SaSc03}.

All in all, at this stage we have ``candidate points'' that may lie in
$\sZ(\cH)$. Writing a quantified formula to decide if there exists a ball around
these points which does not meet $\cH\cap \RR^n$ raises complexity issues
(those points are encoded by zero-dimensional
parametrizations of degree $d^{O(n)}$, given as input to a decision
procedure). 

Hence we need new ingredients. Note that our ``candidate points'' lie on
``curves of critical points'' which are obtained by letting $e$ vary in the
polynomial systems defining the aforementioned critical points. Assume now that
$\cH\cap \RR^n$ is bounded, hence contained in a ball $B$. Then, for $e'$ small
enough, the real algebraic set defined by $f=0$ is ``approximated'' by the union
of the connected components of the real set defined by $f=e'$ which are
contained in $B$. Besides, these ``curves of critical points'', that we just
mentioned, hit these connected components when one fixes $e'$. We actually prove
that two distinct points of our set of ``candidate points'' are connected
through these ``curves of critical points'' and those connected components
defined by $f=e'$ in $B$ if and only if they do not lie in $\sZ(\cH)$. Hence, we
use computations of roadmaps of the real set defined by $f=e'$ to answer those
connectivity queries. Then, advanced algorithms for roadmaps and
polynomial system solving allows us to achieve the announced complexity bound.

Many details are hidden in this description. In particular, we use infinitesimal
deformations and techniques of semi-algebraic geometry. While
infinitesimals are needed for proofs, they may be difficult to use in practice.
On the algorithmic side, we go further exploiting the geometry of the problem to
avoid using infinitesimals.
\vspace{-0.25cm}
\paragraph*{Structure of the paper} In Section \ref{section:geometry}, we study
the geometry of our problem and prove a series of auxiliary results (in
particular Proposition \ref{proposition:main}, which coins the theoretical
ingredient we need). Section \ref{section:algorithm} is devoted to describe the
algorithm. Section \ref{section:complexity} is devoted to the complexity
analysis and Section~\ref{sec:experiments} reports on the
  practical performances of our algorithm.

\noindent
{\em Acknowledgments.} We thank the reviewers for their helpful comments.
\vspace{-0.3cm}

\section{The geometry of the problem}
\label{section:geometry}
\subsection{Candidates for isolated points}
\label{ssection:geocandidates}
As above, let $f \in \QQ[x_1, \ldots, x_n]$ and $\cH\subset \CC^n$ be
the hypersurface defined by $f = 0$. Let $\mathbf{f}$ be a subset of
$\CC[x_1,\ldots,x_n]$, we denote by $\VV(\mathbf{f})$ the simultaneous
vanishing locus in $\CC^n$ of $\mathbf{f}$.

\begin{lemma}\label{lemma:singleton}
  The set $\sZ(\cH)$ is the (finite) union of the semi-algebraically
connected components of $\cH\cap\RR^n$ which are a singleton.
\end{lemma}
\begin{proof}
  Recall that real algebraic sets have a finite number of
semi-algebraically connected components \cite[Theorem 5.21]{BPR}. Let
$\cC$ be a semi-algebraically connected component of $\cH\cap \RR^n$.

  Assume that $C$ is not a singleton and take $\x$ and $\y$ in $\cC$
with $\x\neq \y$. Then, there exists a semi-algebraic continuous map
$\gamma: [0, 1]\to \cC$ s.t.  $\gamma(0)=\x$ and $\gamma(1)=\y$ ;
besides, since $\x\neq \y$, there exist $t\in (0, 1)$ such that
$\gamma(t)\neq \x$. By continuity of $\gamma$ and the norm function,
any ball $B$ centered at $\x$ contains a point $\gamma(t)\neq \x$.

  Now assume that $\cC = \{\x\}$. Observe that $\cH\cap \RR^n -
\{\x\}$ is closed (since semi-algebraically connected components of
real algebraic sets are closed). Since $\cH\cap \RR^n$ is bounded, we
deduce that $\cH\cap \RR^n - \{\x\}$ is closed and bounded. Then, the
map $\y \to \|\y-\x\|^2$ reaches a minimum over $\cH\cap \RR^n -
\{\x\}$.  Let $e$ be this minimum value. We deduce that any ball
centered at $\x$ of radius less than $e$ does not meet $\cH\cap \RR^n
- \{\x\}$.
\end{proof}
To compute those connected components of $\cH\cap\RR^n$ which are
singletons, we use classical objects of optimization and Morse theory
which are mainly {\em polar varieties}. Let $\KK$ be an algebraically
closed field, let $\phi \in \KK[x_1,\ldots,x_n]$ which defines the
polynomial mapping $(x_1, \ldots, x_n) \mapsto \phi(x_1,\ldots,x_n)$
and $V\subset \KK^n$ be a smooth equidimensional algebraic set. We
denote by $W(\phi, V)$ the set of critical points of the restriction
of $\phi$ to $V$. If $c$ is the co-dimension of $V$ and
$(g_1,\ldots,g_s)$ generates the vanishing ideal associated to $V$,
then $W(\phi, V)$ is the subset of $V$ at which the Jacobian matrix
associated to $(g_1, \ldots, g_s, \phi)$ has rank less than or equal
to $c$ (see e.g., \cite[Subsection 3.1]{SaSc17}).

In particular, the case where $\phi$ is replaced by the canonical
projection on the $i$-th coordinate
\[ \pi_i: (x_1, \ldots, x_n) \mapsto x_i,
\] is excessively used throughout our paper.

In our context, we do not assume that $\cH$ is smooth. Hence, to
exploit strong topological properties of polar varieties, we retrieve
a smooth situation using deformation techniques.  We consider an
infinitesimal $\eps$, i.e., a transcendental element over $\RR$ such
that $0 < \eps < r$ for any positive element $r \in \RR$, and the
field of Puiseux series over $\RR$, denoted by
\[ \RR\langle \eps \rangle = \left\{\textstyle\sum_{i\geq
i_0}a_i\eps^{i/q}\mid i \in \bN, i_0\in \bZ, q \in \bN-\{0\}, a_i \in
\RR\right\}.
\]
Recall that $\RR\langle \eps \rangle$ is a real closed
field~\cite[Theorem 2.91]{BPR}. One defines $\CC\langle \eps \rangle$
as for $\RR\langle \eps \rangle$ but taking the coefficients of the
series in $\CC$. Recall that $\CC\langle \eps \rangle$ is an algebraic
closure of $\RR\langle \eps \rangle$~\cite[Theorem
2.17]{BPR}. Consider $\sigma=\sum_{i\geq i_0}a_i\eps^{i/q}\in
\RR\langle \eps \rangle$ with $a_{i_0}\neq 0$. Then, $a_{i_0}$ is
called the \textit{valuation} of $\sigma$. When $i_0\geq 0$, $\sigma$
is said to be {\em bounded over $\RR$} and the set of bounded elements
of $\RR\langle \eps \rangle$ is denoted by $\RR\langle \eps
\rangle_b$. One defines the function $\lim_{\eps}:\RR\langle \eps
\rangle_b \to \RR$ that maps $\sigma$ to $a_0$ (which is $0$ when
$i_0>0$) and writes $\lim_{\eps} \sigma = a_0$; note that
$\lim_{\eps}$ is a ring homomorphism from $\RR\langle \eps \rangle_b$
to $\RR$. All these definitions extend to $\RR\langle \eps \rangle^n$
componentwise. For a semi-algebraic set $\cS \subset \Reps^n$, we
naturally define the limit of $\mathcal{S}$ as $\lim_{\eps}
\mathcal{S} = \left\{\lim_{\eps} \x \; | \; \x \in \mathcal{S}\text{
and }\x \text{ is bounded over }\RR\right\}$.

Let $\cS\subset \RR^n$ be a semi-algebraic set defined by a
semi-algebraic formula $\Phi$. We denote by $\ext(\cS, \RR\langle \eps
\rangle)$ the semi-algebraic set of points which are solutions of
$\Phi$ in $\RR\langle \eps \rangle^n$. We refer to \cite[Chap. 2]{BPR}
for more details on infinitesimals and real Puiseux series.

By e.g., \cite[Lemma 3.5]{RRSa00}, $\cH_\eps$ and
$\cH_{-\eps}$ respectively defined by $f = \eps$ and $f=-\eps$ are two
disjoint smooth algebraic sets in $\Ceps^n$.
\begin{lemma}
\label{lemma:compbounded}
For any $\x$ lying in a bounded connected component of $\cH\cap\RR^n$,
there exists a point $\x_{\eps} \in (\cH_{\eps}\cup \cH_{-\eps})\cap
\Reps_b^n$ such that $\lim_{\eps}\x_{\eps} = \x$. For such a point
$\x_{\eps}$, let $\cC_{\eps}$ be the connected component of
$(\cH_{\eps}\cup \cH_{-\eps})\cap \Reps^n$ containing
$\x_{\eps}$. Then, $\cC_{\eps}$ is bounded over $\RR$.
\end{lemma}
\begin{proof}
  See \cite[Lemma 3.6]{RRSa00} for the first claim. The second part can
  be deduced following the proof of \cite[Proposition 12.51]{BPR}.
\end{proof}
\begin{proposition}\label{prop:Ceps} Assume that $\sZ(\cH)$ is not
empty and let $\x\in \sZ(\cH)$. There exists a semi-algebraically
connected component $\cC_\eps$ that is bounded over $\RR$ of
$(\cH_\eps\cup \cH_{-\eps})\cap \Reps^n$ such that $\lim_{\eps}
\cC_{\eps} = \{\x\}$.

Consequently, for $1\leq i \leq n$, there exists an $\x_\eps\in
(W(\pi_i, \cH_\eps)\cup W(\pi_i, \cH_{-\eps}))\cap \cC_\eps$ such that
$\lim_{\eps} \x_\eps = \x$. Hence we have that \[\textstyle{\sZ(\cH) \subset
\cap_{i=1}^n \lim_{\eps} ((W(\pi_i,\cH_{\eps})\cup
W(\pi_i,\cH_{-\eps}))\cap \Reps^n_b).}\]
\end{proposition}

\begin{proof}
By Lemma \ref{lemma:compbounded}, there exists $\x_{\eps} \in
(\cH_{\eps}\cup \cH_{-\eps})\cap\Reps^n$ such that $\lim_{\eps}
\x_{\eps} = \x$. Assume that $\x_{\eps} \in \cH_{\eps}$ and let
$\cC_{\eps}$ be the connected component of $\cH_{\eps}\cap \Reps^n$
containing $\x_{\eps}$. Again, by Lemma \ref{lemma:compbounded},
$\cC_{\eps}$ is bounded over $\RR$. We prove that $\lim_{\eps}
\cC_\eps = \{\x\}$ by contradiction. The case $\x_{\eps} \in
\cH_{-\eps}$ is done similarly.

Assume that there exists a point $\y_{\eps}\in \cC_\eps$ such that
$\lim_{\eps} \y_{\eps} = \y$ and $\y \ne \x$. Since $\cC_\eps$ is
semi-algebraically connected, there exists a semi-algebraically
continuous function $\gamma:\ext([0, 1],\Reps) \to \cC_\eps$ such that
$\gamma(0)=\x_{\eps}$ and $\gamma(1)=\y_{\eps}$. By \cite[Proposition
12.49]{BPR}, $\lim_{\eps} {\rm Im}(\gamma)$ is connected and contains
$\x$ and $\y$. As $\lim_{\eps}$ is a ring homomorphism, $f(\lim_{\eps}
\gamma(t)) = \lim_{\eps} f(\gamma(t)) = 0$, so $\lim_{\eps} {\rm
Im}(\gamma)$ is contained in $\cH\cap \RR^n$. This contradicts the
isolatedness of $\x$, then we conclude that $\lim_{\eps}\cC_{\eps} =
\{\x\}$.

Since $\cC_\eps$ is a semi-algebraically connected component of the
real algebraic set $\cH_\eps\cap\RR\langle \eps \rangle^n$, it is
closed. Also, $\cC_{\eps}$ is bounded over $\RR$. Hence, for any
$1\leq i \leq n$, the projection $\pi_i$ reaches its extrema over
$\cC_\eps$ \cite[Proposition 7.6]{BPR}, which implies that
$\cC_\eps\cap W(\pi_i, \cH_\eps)$ is non-empty. Take $\x_\eps\in
W(\pi_i,\cH_\eps)\cap \cC_\eps$, then $\x_{\eps}$ is bounded over
$\RR$ and its limit is $\x$. Thus, $\sZ(\cH) \subset \lim_{\eps}
(W(\pi_i,\cH_{\eps}) \cap \Reps_b^n)$ for any $1 \leq i\leq n$, which
implies $\sZ(\cH) \subset \cap_{i=1}^n \lim_{\eps}
(W(\pi_i,\cH_{\eps})\cap \Reps^n_b)$.
\end{proof}
\subsection{Simplification}
\label{ssection:non-compact}
We introduce in this subsection a method to reduce our problem to the
case where $\cH\cap \RR^n$ is bounded for all $\x \in \RR^n$. Such
assumptions are required to prove the results in Subsection
\ref{ssection:reciprocal}. Our technique is inspired by \cite[Section
12.6]{BPR}. The idea is to associate to the possibly unbounded
algebraic set $\cH\cap \RR^n$ a bounded real algebraic set whose
isolated points are strongly related to $\sZ(\cH)$. The construction
of such an algebraic set is as follows.

Let $x_{n+1}$ be a new variable and $0<\rho\in \RR$ such that $\rho$ is
  greater than the Euclidean norm $\|\cdot\|$ of every isolated point of
  $\cH\cap \RR^n$. Note that such a $\rho$ can be obtained from a finite set of
  points containing the the isolated points of $\cH\cap\RR^n$. We explain in
  Subsection~\ref{ssection:limit} how to compute such a finite set.

We consider the algebraic set $\cV$
defined by the system
\[f=0,\; x_1^2+\ldots+x_n^2+x_{n+1}^2-\rho^2 = 0.\] Let $\pi_{\x}$ be
the projection $(x_1,\ldots,x_n,x_{n+1}) \mapsto (x_1,\ldots,x_n)$.


The real counterpart of $\cV$ is the intersection of $\cH$
lifted to $\RR^{n+1}$ with the sphere of center $\mathbf{0}$ and
radius $\rho$. Therefore, $\cV$ is a bounded real algebraic
set in $\RR^{n+1}$. Moreover, the restriction of $\pi_{\x}$ to
$\cV\cap \RR^{n+1}$ is exactly $\cH\cap B(\mathbf{0},\rho)$. By
the definition of $\rho$, this image contains all the real isolated
points of $\cH$. Lemma \ref{lemma:non-compact} below relates
$\sZ(\cH)$ to the isolated points of $\cV\cap\RR^{n+1}$.
\begin{lemma}
Let $\cV$ and $\pi_{\x}$ as above. We denote by $\sZ(\cV) \subset
\RR^{n+1}$ the set of real isolated points of $\cV$ with non-zero
$x_{n+1}$ coordinate. Then, $\pi_{\x}(\sZ(\cV))=\sZ(\cH)$.
\label{lemma:non-compact}
\end{lemma}
\begin{proof}
Note that $\pi_{\x}(\cV \cap \RR^{n+1}) = (\cH\cap \RR^n) \cap
B(\mathbf{0},\rho)$.
We consider a real isolated point
$\x'=(\alpha_1,\ldots,\alpha_n,\alpha_{n+1})$ of $\cV$ with
$\alpha_{n+1}\ne 0$ and $\x = \pi_{\x}(\x')=
(\alpha_1,\ldots,\alpha_n)$. Assume by contradiction that $\x \not \in
\sZ(\cH)$, we will prove that $\x' \not \in \sZ(\cV)$, i.e.,
for any $r>0$, there exists
$\y'=(\beta_1,\ldots,\beta_n,\beta_{n+1})\in \cV \cap
\RR^{n+1}$ such that $\| \y' - \x'\| < r$.
Since $\x$ is not isolated, there exists a point $\y \ne \x$ such that
$\|\y - \x\| < \frac{r}{1+2\rho/|\alpha_{n+1}|}$. Let $\y' \in
\pi_{\x}^{-1}(\y)$ such that $\alpha_{n+1}\beta_{n+1} \ge 0$. We have
that $\|\x\|^2+\alpha_{n+1}^2 = \|\y\|^2+\beta_{n+1}^2 =
\rho^2$. Now we estimate
\[ |\|\y\|^2 - \|\x\|^2| = (\|\x\|+\|\y\|)\cdot |\|\y\| -
\|\x\|| \leq 2 \rho \cdot \|\y-\x\|,\]
\[ |\alpha_{n+1}-\beta_{n+1}| \leq
\frac{|\alpha_{n+1}^2-\beta_{n+1}^2|}{|\alpha_{n+1}|} =
\frac{|\|\y\|^2 - \|\x\|^2|}{|\alpha_{n+1}|} \leq \frac{2\rho\cdot
\|\y-\x\|}{|\alpha_{n+1}|}.\]
Finally,
\[ \|\y'-\x'\| \leq \|\y-\x\|+|\alpha_{n+1}-\alpha_{n+1}| \leq
\left(1+\frac{2\rho}{|\alpha_{n+1}|}\right) \| \y - \x \| < r.\]
So, $\x'$ is not isolated in $\cV\cap \RR^{n+1}$. This contradiction
implies that $\pi_{\x}(\sZ(\cV))\subset\sZ(\cH)$.

It remains to prove that $\sZ(\cH)\subset \pi_{\x}(\sZ(\cV))$. For any
real isolated point $\x \in \sZ(\cH)$, we consider a ball $B(\x, r')
\subset B(\mathbf{0},\rho) \subset \RR^n$ such that $B(\x,r') \cap \cH
= \{\x\}$. We have that $\pi_{\x}^{-1}(B(\x,r'))\cap \cV\cap
\RR^{n+1}$ is equal to $\pi_{\x}^{-1}(\x)\cap \cV\cap \RR^{n+1}$,
which is finite. So, all the points in $\pi_{\x}^{-1}(B(\x,r'))\cap
\cV\cap \RR^{n+1}$ are isolated. Since $\sZ(\cH) \subset
B(\mathbf{0},\rho)$, we deduce that $\sZ(\cH)$ is contained in
$\pi_{\x}(\sZ(\cV))$.

Thus, we conclude that $\pi_{\x}(\sZ(\cV))=\sZ(\cH)$.
\end{proof}


Note that the condition $x_{n+1}\ne 0$ is crucial. For a connected
component $\mathcal{C}$ of $\cH\cap \RR^n$ that is not a singleton,
its intersection with the closed ball $\overline{B(\mathbf{0},\rho)}$
can have an isolated point on the boundary of the ball, which
corresponds to an isolated point of $\cV\cap\RR^{n+1}$. This situation
depends on the choice of $\rho$ and can be easily detected by checking
the vanishing of the coordinate $x_{n+1}$.
\subsection{Identification of isolated points}
\label{ssection:reciprocal}
By Proposition \ref{prop:Ceps}, the real points of $\cap_{i=1}^n
\lim_{\eps}W(\pi_i,\cH_{\eps})$ are potential isolated points of
$\cH\cap \RR^n$.  We study now how to identify, among those
candidates, which points are truely isolated.

We use the same $g = x_1^2+\ldots+x_{n+1}^2-\rho^2$ and $\cV =\VV(f,g)
\subset \CC^{n+1}$ as in Subsection \ref{ssection:non-compact}. Let
$\cV_{\eps} = \VV(f-\eps,g)$ and $\cV_{-\eps} = \VV(f+\eps,g)$, note
that they are both algebraic subsets of $\Ceps^{n+1}$.

\begin{lemma}\label{lemma:existence}
  Let $\x \in \cV \cap \RR^{n+1}$ such that its $x_{n+1}$ coordinate
is non-zero. Then, $\x$ is not an isolated point of $\cV\cap\RR^{n+1}$
if and only if there exists a semi-algebraically connected component
$\cC_{\eps}$ of $(\cV_{\eps}\cup \cV_{-\eps})\cap \Reps^{n+1}$,
bounded over $\RR$, such that $ \{\x\} \subsetneq \lim_{\eps}
\cC_{\eps}$.
\end{lemma}
\begin{proof}
Let $\x = (\alpha_1,\ldots,\alpha_{n+1}) \in \cV \cap \RR^{n+1}$ such
that $\alpha_{n+1} \ne 0$. As $f(\alpha_1,\ldots,\alpha_n)=0$, by Lemma
\ref{lemma:compbounded}, there exists a point
$\x_{\eps}=(\beta_1,\ldots,\beta_{n+1}) \in \Reps^{n+1}$ such that
$(\beta_1,\ldots,\beta_n) \in (\cH_{\eps}\cup \cH_{-\eps})\cap\Reps^n$
and $\lim_\eps (\beta_1,\ldots,\beta_n) =
(\alpha_1,\ldots,\alpha_n)$. Since $\alpha_{n+1} \ne 0$, we can choose
$\beta_{n+1}$ such that $g(\x_{\eps}) = 0$. Therefore, for any $\x$ as
above, there exists $\x_{\eps} \in
(\cV_{\eps}\cup\cV_{-\eps})\cap\Reps^{n+1}$ such that $\lim_{\eps}
\x_{\eps} = \x$.

Since $(\cV_{\eps}\cup\cV_{-\eps})\cap \Reps^{n+1}$ lies on the sphere
(in $\Reps^{n+1}$) defined by $g = 0$, every connected component of
$(\cV_{\eps}\cup\cV_{-\eps})\cap\Reps^{n+1}$ is bounded over
$\RR$. Hence, the points of $\cV\cap \RR^{n+1}$ whose $x_{n+1}$
coordinates are not zero are contained in
$\lim_{\eps}(\cV_{\eps}\cup\cV_{-\eps})\cap\Reps^{n+1}$.

Let $\x$ be a non-isolated point of $\cV\cap \RR^{n+1}$ whose
$x_{n+1}$-coordinate is not zero. We assume by contradiction that for
any semi-al\-ge\-bra\-i\-cal\-ly connected component $\cC_{\eps}$ of
$(\cV_{\eps}\cup \cV_{-\eps})\cap \Reps^{n+1}$ (which is bounded over $\RR$ by above),
then it happens that either $\lim_{\eps} \cC_{\eps} = \{\x\}$ or $\x
\not \in \lim_{\eps}\cC_{\eps}$.

Since $(\cV_{\eps}\cup\cV_{-\eps})\cap\Reps^{n+1}$ has finitely many
connected components, the
number of connected components of the second type is also
finite. Since $\cV\cap\RR^{n+1}$ is not a singleton (by the existence
of $\x$), the connected components of the second type exist. So, we
enumerate them as $\cC_1,\ldots, \cC_k$ and $\x \not \in \lim_{\eps}
\cC_j$ for $1\leq j \leq k$.

As $\x$ is not isolated in $\cV\cap \RR^{n+1}$ with non-zero $x_{n+1}$
coordinate by assumption, there exists a sequence of points
$(\x_i)_{i\ge 0}$ in $\cV\cap \RR^{n+1}$ of non-zero $x_{n+1}$
coordinates that converges to $\x$.  Since there are finitely many
$\cC_i$, there exists an index $j$ such that $\lim_{\eps} \cC_j$
contains a sub-sequence of $(\x_i)_{i\ge 0}$. By Proposition 12.49
[BPR], the limit of the semi-algebraically connected component
$\cC_{j}$ (which is bounded over $\RR$) is a closed and connected
semi-algebraic set. It follows that $\x \in \lim_{\eps}\cC_j$, which
is a contradiction. Therefore, there exists a semi-algebraically
connected component of $(\cV_{\eps}\cup\cV_{-\eps})\cap\Reps^{n+1}$, bounded over $\RR$,
such that $\{\x\}\subsetneq \lim_{\eps} \cC_{\eps}$.

It remains to prove the reverse implication. Assume that $\{\x\}
\subsetneq \lim_{\eps} \cC_{\eps}$ for some semi-algebraically
connected component $\cC_{\eps}$ of $(\cV_\eps\cup\cV_{-\eps})\cap \Reps^{n+1}$ that is
bounded over $\RR$. As $\lim_{\eps} \cC_{\eps}$ is connected, we
finish the proof.
\end{proof}
\begin{lemma}
   Let $\x \in \cV\cap \RR^{n+1}$ whose $x_{n+1}$ coordinate is
non-zero. Assume that $\x$ is not an isolated point of $\cV\cap
\RR^{n+1}$. For any semi-algebraically connected component
$\cC_{\eps}$ of $(\cV_{\eps}\cup\cV_{-\eps})\cap\RR^{n+1}$, bounded
over $\RR$, such that $\{\x\}\subsetneq \lim_{\eps} \cC_{\eps}$, there
exists $1\leq i\leq n$ such that $\cC_{\eps}\cap
(W(\pi_i,\cV_{\eps})\cup W(\pi_i,\cV_{-\eps}))$ contains a point
$\x'_{\eps}$ which satisfies $\lim_{\eps}\x'_{\eps} \ne \x$.
\label{lemma:critical}
\end{lemma}
\begin{proof}
  Let $\cC_{\eps}$ be semi-algebraically connected component of
$(\cV_{\eps}\cup\cV_{-\eps})\cap\Reps^{n+1}$, bounded over $\RR$, such
that $\{\x\}\subsetneq \lim_{\eps} \cC_{\eps}$. Lemma
\ref{lemma:existence} ensures the existence of such a connected
component $\cC_{\eps}$.

Now let $\x_{\eps}$ and $\y_{\eps}$ be two points contained in
$\cC_{\eps}$ such that $\lim_{\eps} \x_{\eps} = \x$,
$\lim_{\eps}\y_{\eps}=\y$ and $\x\ne\y$. Let
$\x=(\alpha_1,\ldots,\alpha_{n+1})$ and
$\y=(\beta_1,\ldots,\beta_{n+1})$. Since $\x \ne \y$, there exists
$1\leq i \leq n+1$ such that $\alpha_i \ne \beta_i$. Note that if
$(\alpha_1,\ldots,\alpha_n)=(\beta_1,\ldots,\beta_n)$ for any $\y\in
\lim_{\eps}\cC_{\eps}$, then $\lim_{\eps}\cC_{\eps}$ contains at most
two points (by the constraint $g=0$). However, since
$\lim_{\eps}\cC_{\eps}$ is connected and contains at least two points,
it must be an infinite set. So, we can choose $\y$ such that have that
$1\leq i \leq n$.

As $\cC_{\eps}$ is closed in $\Reps^{n+1}$ (as a connected component
of an algebraic set) and bounded over $\RR$ by definition, its
projection on the $x_i$ coordinate is a closed interval $[a,b] \subset
\Reps$ (see \cite[Theorem 3.23]{BPR}), which is bounded over $\RR$
(because $\cC_{\eps}$ is). Also, since $[a,b]$ is closed, there exist
$\x'_a$ and $\x'_b$ in $\Reps^{n+1}$ such that $\x'_a \in
\pi_i^{-1}(a)\cap \cC_{\eps}\cap (W(\pi_i, \cV_{\eps})\cup
W(\pi_i,\cV_{-\eps}))$ and $\x'_b\in \pi_i^{-1}(b)\cap \cC_{\eps}\cap
(W(\pi_i, \cV_{\eps})\cup W(\pi_i,\cV_{-\eps}))$.  Since $\alpha_i\ne
\beta_i$ both lying in $\RR$, $\{\alpha_i, \beta_i\} \subset
[\lim_{\eps} a, \lim_{\eps} b]$ implies that $\lim_{\eps} a \ne
\lim_{\eps} b$. It follows that $\lim_{\eps} \x'_a \ne \lim_{\eps}
\x'_b$. Thus, at least one point among $\lim_{\eps} \x'_a$ and
$\lim_{\eps} \x'_b$ does not coincide with $\x$. Hence, there exists a
point $\x'_{\eps}$ in $\cC_{\eps}\cap (W(\pi_i,\cV_{\eps})\cup
W(\pi_i,\cV_{-\eps}))$ such that $\lim_{\eps} \x'_{\eps} \ne \x$.
\end{proof}
We can easily deduce from Lemma \ref{lemma:existence} and Lemma
\ref{lemma:critical} the following proposition, which is the main
ingredient of our algorithm.
\begin{proposition}
\label{proposition:main}
Let $\x\in \cap_{i=1}^n \lim_{\eps} W(\pi_i, \cV_{\eps})\cup
W(\pi_i,\cV_{-\eps})$ whose $x_{n+1}$ coordinate is non-zero. Then,
$\x$ is not an isolated point of $\cV\cap\RR^{n+1}$ if and only if
there exist $1\leq i \leq n$ and a connected component
$\cC_\eps$ of $\cV_{\eps}\cap \Reps^{n+1}$, which is bounded over $\RR$,
such that $\cC_\eps\cap W(\pi_i,\cH_{\eps})$ contains $\x_{\eps}$,
$\x'_{\eps}$ satisfying $\x=\lim_{\eps} \x_{\eps} \ne \lim_\eps
\x'_\eps$.
\end{proposition}
\section{Algorithm}
\label{section:algorithm}
\subsection{General description}
\label{ssection:geodescription}
The algorithm takes as input a polynomial $f\in \RR[x_1, \ldots, x_n]$. 


The first step consists in computing a parametrization $\kP$
encoding a finite set of points which contains $\sZ(\cH)$. Let
$\cH_\eps$ and $\cH_{-\eps}$ be the algebraic subsets of $\Ceps^n$
respectively defined by $f=\eps$ and $f=-\eps$. By
Proposition \ref{prop:Ceps}, the set $\cap_{i=1}^n \lim_{\eps}
W(\pi_i,\cH_{\eps})\cup W(\pi_i,\cH_{-\eps})$ contains the
real isolated points of $\cH$. To ensure that this set is finite, we
use {\em generically chosen} linear change of coordinates.

Given a matrix $\bmA \in GL_{n}(\QQ)$, a polynomial $p\in
\QQ[x_1,\ldots,x_{n}]$ and an algebraic set $\mathcal{S} \subset
\CC^{n}$, we denote by $p^{\bmA}$ the polynomial $p(\bmA\cdot \x)$
obtained by applying the change of variables $\bmA$ to $p$ and
$\mathcal{S}^\bmA = \{\bmA^{-1}\cdot\x\;|\;\x \in
\mathcal{S}\}$. Then, we have that $\VV(p)^\bmA = \VV(p^\bmA)$.

In \cite{SaSc03}, it is proved that, with $\bmA$ outside a prescribed
proper Zariski closed subset of $GL_{n}(\QQ)$,
$W(\pi_i,\cH_{\eps}^{\bmA})\cup W(\pi_i,\cH_{-\eps}^{\bmA})$ is finite for $1\leq i \leq
n$. Additionally, since $\bmA$ is assumed to be generically chosen,
\cite{Sa05} shows that the ideal $\left\langle \ell\cdot
\diff{f^{\bmA}}{x_i}-1,\diff{f^{\bmA}}{x_j} \text{ for all }j\ne
i\right\rangle$ defines either an empty set or a one-equidimensional
algebraic set, where $\ell$ is a new variable. Those extra assumptions
are required in our subroutine {\sf Candidates} (see the next
subsection). Note that, for any matrix $\bmA$, the real isolated points
of $\cH^\bmA$ is the image of $\sZ(\cH)$ by the linear mapping
associated to $\bmA^{-1}$. Thus, in practice, we will choose randomly
a $\bmA \in GL_{n+1}(\QQ)$, compute the real isolated points of
$\cH^\bmA$, and then go back to $\sZ(\cH)$ by applying the change of
coordinates induced by $\bmA^{-1}$. This random choice of $\bmA$ makes
our algorithm probabilistic.

The next step consists of identifying those of the candidates which
are isolated in $\cH^\bmA\cap\RR^n$ ; this step relies on Proposition
\ref{proposition:main}. To reduce our problem to the context where
Proposition \ref{proposition:main} can be applied, we use
Lemma~\ref{lemma:non-compact}. One needs to compute $\rho\in \RR$,
such that $\rho$ is larger than the maximum norm of the real isolated
points we want to compute. This value of $\rho$ can be easily obtained
by isolating the real roots of the zero-dimensional parametrization
encoding the candidates. Further, we call {\sf GetNormBound} a
subroutine which takes as input $\kP$ and returns $\rho$ as we just
sketched. We let $g = x_1^2+\ldots+x_n^2+x_{n+1}^2-\rho^2$. By Lemma
\ref{lemma:non-compact}, $\sZ(\cH)$ is the projection of the set of real
isolated points of the algebraic set $\cV$ defined by $f=g = 0$ at
which $x_{n+1}\neq 0$. Let $\mathcal{X}$ be the set of points of $\cV$
projecting to the candidates encoded by $\kP$.

Proposition \ref{proposition:main} would lead us to compute $W(\pi_i,
\cV_\eps^\bmA)\cup W(\pi_i,\cV_{-\eps}^\bmA)$ as well as a roadmap of $\cV^\bmA_\eps\cup
\cV^\bmA_{-\eps}$. As explained in the introduction, this induces
computations over the ground field $\Reps$ which we want to avoid. We
bypass this computational difficulty as follows. We compute a roadmap
$\cR_e$ for $\cV_e^\bmA\cup \cV_{-e}^\bmA\cap \RR^{n+1}$ (defined by
$\{f^{\bmA} = e,g = 0\}$ and $\{f^\bmA=-e,g=0\}$ respectively) for $e$
small enough (see Subsection \ref{ssection:connectivity})
and define a semi-algebraic curve $\cK$ containing $\mathcal{X}$ such
that $\x\in \mathcal{X}$ is isolated in $\cV^\bmA\cap\RR^{n+1}$ if and
only if it is not connected to any other $\x'\in \mathcal{X}$ by
$\cK$. We call {\sf IsIsolated} the subroutine that takes as input
$\kP$, $f^\bmA$ and $g$ and returns $\kP$ with isolating boxes $\BB$
of the real points of defined by $\kP$ which are isolated in
$\cV^\bmA\cap \RR^{n+1}$.

Once the real isolated points of $\cV^\bmA$ is computed, we remove the
boxes corresponding to points at which $x_{n+1} = 0$ and project the
remaining points on the $(x_1, \ldots, x_n)$-space to obtain the
isolated points of $\cH^\bmA$. This whole step uses a subroutine which
we call {\sf Remove} (see \cite[Appendix J]{SaSc17}). Finally, we
reverse the change of variable by applying $\bmA^{-1}$ to get $\sZ(\cH)$.

We summarize our discussion in Algorithm \ref{algo:main} below.

\begin{algorithm}
\KwData{A polynomial $f \in \QQ[x_1,\ldots,x_n]$}
\KwResult{A zero-dimensional parametrization $\kP$ such that $\sZ(\cH) \subset \ZZ(\kP)$ and a set of boxes isolating $\sZ(\cH)$}
\SetAlgoNoLine
$\bmA$ chosen randomly in $GL_{n+1}(\QQ)$\\
$\kP \gets \text{{\sf Candidates}}(f^\bmA)$\\
$\rho \gets {\sf GetNormBound}(\kP)$\\
$g \gets x_1^2+\ldots+x_n^2+x_{n+1}^2-\rho^2$ \\
$\kP, B \gets \text{IsIsolated}(\kP,f^\bmA,g)$ \\
$\kP, B \gets \text{Removes}(\kP,B, x_{n+1})$ \\
$\kP, B \gets \kP^{\bmA^{-1}}, B^{\bmA^{-1}}$ \\
\Return $(\kP, B)$
\caption{{\sf IsolatedPoints}}
\label{algo:main}
\end{algorithm}

\subsection{Computation of candidates}
\label{ssection:limit}
Further, we let $\cH^\bmA_\eps$ (resp. $\cH^\bmA_{-\eps}$) be the
algebraic set associated to $f^\bmA=\eps$ (resp. $f^\bmA=-\eps$).  To
avoid to overload notation, we omit the change of variables $\bmA$ as
upper script. Let $\ell$ be a new variable. For $1\leq i \leq n$,
$I_i$ denotes the ideal of $\QQ[\ell,x_1,\ldots,x_{n}]$ generated by
the set of polynomials $\left\{\ell\cdot \diff{f}{x_i}-1,
\diff{f}{x_j}\quad \text{ for all } j \ne i\right\}$.

Following the discussion in
Subsection \ref{ssection:geodescription}, the algebraic set associated
to $I_i$ is either empty or one-equidimensional and
$W(\pi_i,\cH_{\eps})\cup W(\pi_i,\cH_{-\eps})$ is finite. Hence,
\cite[Theorem 1]{Sa05} shows that the algebraic set associated to the
ideal $\left\langle f \right\rangle + (I_i\cap \QQ[x_1,\ldots,x_n])$
is zero-dimensional and contains $\lim_{\eps} W(\pi_i,\cH_\eps) \cup
W(\pi_i,\cH_{-\eps})$.

In our problem, the intersection of $\lim_{\eps}
W(\pi_i,\cH_{\eps})\cup W(\pi_i,\cH_{-\eps})$ is needed rather
than each limit itself. Hence, we use the inclusion
\[\cap_{i=1}^{n} \lim_{\eps} W(\pi_i,\cH_{\eps}) \cup
  W(\pi_i,\cH_{-\eps}) \subset
\VV\left(\left\langle f \right\rangle +
\textstyle\sum_{i=1}^{n}I_i\cap \QQ[x_1,\ldots,x_{n}]\right).\]
We can compute the algebraic set on the right-hand side as follows:
\begin{enumerate}
\item For each $1\leq i \leq n$, compute a set $G_i$ of generators of the ideal $I_i\cap \QQ[x_1,\ldots,x_n]$.
\item Compute a zero-dimensional parametrization $\kP$ of the set of
  polynomials $\left\{f\right\} \cup G_1 \cup \ldots \cup G_n$.
\end{enumerate}
Such computations mimic those in \cite{Sa05}. The complexity of this
algorithm of course depends on the algebraic elimination procedure we
use.  For the complexity analysis in Section \ref{section:complexity},
we employ the geometric resolution \cite{GLS01}. It basically consists
in computing a one-dimensional parametrization of the curve defined by
$I_i$ and next computes a zero-dimensio\-nal parametrization of the
finite set obtained by intersecting this curve with the hypersurface
defined by $f = 0$. We call {\sf ParametricCurve} a subroutine that,
taking the polynomial $f$ and $1\leq i\leq n$, computes a
one-dimensional parametrization $\kG_i$ of the curve defined
above. Also, let {\sf IntersectCurve} be a subroutine that, given a
one-dimensional rational parametrization $\kG_i$ and $f$, outputs a
zero-dimensional parametrization $\kP_i$ of their
intersection. Finally, we use a subroutine {\sf Intersection} that,
from the parametrizations $\kP_i$'s, computes a zero-dimensional
parametrization of $\cap_{i=1}^{n} \ZZ(\kP_i)$.

\begin{algorithm}
\KwData{The polynomial $f \in \QQ[x_1,\ldots,x_{n}]$}
\KwResult{A zero-dimensional parametrization $\kP$}
\SetAlgoNoLine
\For{$1 \leq i \leq n$}{
$\textstyle{\kG_i \gets \text{\sf ParametricCurve}(\nf,i)}$ \\
$\textstyle{\kP_i \gets \text{\sf IntersectCurve}(\kG_i,\nf)}$ \\
}
$\textstyle{\kP \gets \text{\sf Intersection}(\kP_1,\ldots,\kP_{n}})$ \\
\Return  $\kP$
\caption{Algorithm {\sf Candidates}}
\label{algo:candidates}
\end{algorithm}

\subsection{Description of {\sf IsIsolated}}
\label{ssection:connectivity}
This subsection is devoted to the subroutine {\sf IsIsolated} that
identifies isolated points of $\cH^\bmA\cap \RR^n$ among the
candidates $\ZZ(\kP) \cap \RR^n$ computed in the previous
subsection. We keep using $f$ to address $f^\bmA$. Let $\cP$ be the
set $\{\x = (x_1,\ldots,x_{n+1}) \in \RR^{n+1}\; | \;
(x_1,\ldots,x_n)\in \ZZ(\kP), g(\x) = 0, x_{n+1}\ne 0\}$.

For $e\in \RR$, $\cV_e\in \CC^{n+1}$ denotes the algebraic set
defined by $f = e$ and $g=0$. We follow the idea mentioned in the end
of Subsection \ref{ssection:geodescription}, that is to replace the
infinitesimal $\eps$ by a sufficiently small $e\in \RR$ then adapt the
results of Subsection \ref{ssection:reciprocal} to $\cV_e$.

By definition, $\cV_e\cap \RR^{n+1}$ is bounded for any $e\in
\RR$. Let $t$ be a new variable, $\pi_{\x}:
(\x,t) \mapsto \x$ and $\pi_t: (\x,t) \mapsto t$. For a semi-algebraic
set $\cS \subset \RR^{n+1}\times \RR$ in the coordinate $(\x,t)$ and a
subset $\mathcal{I}$ of $\RR$, the notation $\cS_{\mathcal{I}}$ stands
for the fiber $\pi_t^{-1}(\mathcal{I})\cap \cS$.
Let $\cV_t = \{(\x,t)\in \RR^{n+1}\times \RR \; | \;
f(\x)=t,g(\x)=0\}$. Note that $\cV_t$ is smooth.  Recall
that the set of critical values of the restriction of $\pi_t$ to
$\cV_t$ is finite by the algebraic Sard's theorem (see e.g.,
\cite[Proposition B.2]{SaSc17}).

Since for $e\in \RR$, the set $\cV_e\cap \RR^{n+1}$ is
compact, the restriction of  $\pi_t$ to $\cV_t$ is proper.
Then, by Thom's isotopy lemma \cite{CoSh92}, $\pi_t$ realizes a
locally trivial fibration over any open connected subset of $\RR$
which does not intersect the set of critical values of the restriction
of $\pi_t$ to $\cV_t$. Let $\eta \in \RR$ such that the open set
$]-\eta,0[ \cup ]0,\eta[$ does not contain any critical value of the
restriction of $\pi_t$ to the algebraic set $\cV_t$. Hence, $\cV_e$ is
nonsingular for $e \in ]-\eta,0[ \cup ]0,\eta[$, $(\cV_e\cap
\RR^{n+1})\times (]-\eta,0[ \cup ]0, \eta[)$ is diffeomorphic to
$\cV_{t,]-\eta,0[\cup ]0,\eta[}$.

We need to mention that $W(\pi_i,\cH_e)$ corresponds to the critical
points of $\pi_i$ restricted to $\cV_e$ with non-zero $x_{n+1}$
coordinate. Further, we use $W(\pi_i,\cV_e)$ to address those latter
critical points.

Now, for $1\leq i \leq n$, we define $\cW_i$ as the closure of
\[\textstyle{\left\{(\x,t) \in \RR^{n+2}\; |
\;\diff{f}{x_i}(\x)\neq 0, \diff{f}{x_j}(\x) = 0\text{ for }
j\ne i,x_{n+1}\ne 0\right\}\cap \cV_t}.\]
Since $\bmA$ is assumed to be generically chosen,
$\cW_i$ is either empty or one-equidimensional (because
$\textstyle{\langle \ell\cdot \diff{f}{x_i}-1,\diff{f}{x_j}
\forall j\ne i\rangle}$ either defines an empty set or a
one-equidimensional algebraic set by \cite{Sa05}). This implies that
the set of singular points of $\cW_i$ is finite.

By \cite{JK05}, the set of non-properness of the restriction of
$\pi_t$ to $\cW_i$ is finite (this is the set of points $y$ such that
for any closed interval $U$ containing $y$, $\pi_t^{-1}(U)\cap \cW_i$ is
not bounded). Using again \cite{JK05}, the restriction of $\pi_t$ to
$\cW_i$ realizes a locally trivial fibration over any connected open
subset which does not meet the union of the images by $\pi_t$ of the
singular points of $\cW_i$, the set of non-properness, and the set of
critical values of the restriction of $\pi_t$ to $\cW_i$. We let
$\eta'_i$ be the minimum of the absolute values of the points in this
union.

We choose now $0 < e_0 < \min(\eta, \eta'_1,\ldots,\eta'_n)$. We call
{\sf SpecializationValue} a subroutine that takes as input $f$ and $g$
and returns such a rational number $e_0$. Note that {\sf
SpecializationValue} is easily obtained from elimination algorithms
solving polynomial systems (from which we can compute critical values)
and from \cite{SaSc04} to compute the set of non-properness of some
map.

With $e_0$ as above, we denote $\cI = ]-e_0,0[\cup ]0,e_0[$. Let
$\cW_{i,\cI}$ is semi-algebraically diffeomorphic to $\cW_{i,e} \times
\cI$ for every $ e \in \cI$. As $\cV_e$ is
nonsingular, the critical locus $W(\pi_i,\cV_e)$ is guaranteed to
be finite by the genericity of the change of variables $\bmA$ (hence
$\cW_{i,e}$ is) and that $W(\pi_i,\cV_e)\cap \RR^{n+1}$ coincides
with $\pi_{\x}(\cW_{i,e})$. Thus, the above diffeomorphism implies
that, for any connected component $\cC$ of $\cW_{i,\cI}$, $\cC$ is
diffeomorphic to an open interval in $\RR$. Moreover, if $\cC$ is
bounded, then $\overline{\cC}\setminus \cC$ contains exactly two
points which satisfy respectively $f=0$ and $f^2=e_0^2$. We now
consider
\[\textstyle{\cL_i = \left\{\x \in \RR^{n+1} \; | \; 0 < f < e_0,
g = 0, \diff{f}{x_j}=0 \text{ for} j \ne i ,x_{n+1}\ne
0\right\}}.\]
It is the intersection of the Zariski closure $\cK_i$ of
the solution set to
$\left\{\diff{f}{x_i}\ne 0, \diff{f}{x_j}=0 \text{ for } j \ne
i,x_{n+1}\ne 0\right\}$
with the semi-algebraic set defined by $0 < f < e_0$. Note that
$\cK_i$ is either empty or one-equidimensional. As $\cV_e$ is
nonsingular for $e\in\cI$, $\cL_i$ and $\cL_j$ are
disjoint for $i\ne j$.
Since the restriction of $\pi_{\x}$ to $\cV_t$ is an isomorphism
between the algebraic sets $\cV_t$ and $\RR^{n+1}$ with the inverse
map $\x \mapsto (\x,f(\x))$, the properties of $\cW_{\cI}$ mentioned
above are transferred to its image $\cL_i$ by the projection
$\pi_{\x}$.

Further, we consider a subroutine {\sf ParametricCurve} which
takes as input $f$ and $i\in [1,n]$ and returns a rational
parametrization $\mathfrak{K}_i$ of $\cK_i$. Also, let {\sf Union} be
a subroutine that takes a family of rational parametrizations
$\mathfrak{K}_1, \ldots, \mathfrak{K}_n$ to compute a rational
parametrization encoding the union of the algebraic curves defined by
the $\mathfrak{K}_i$'s. We denote by $\kK$ the output of {\sf Union} ;
it encodes $\cK = \cup_{i=1}^n\cK_i$. We refer to \cite[Appendix
J.2]{SaSc17} for these two subroutines.

Lemma \ref{lemma:elimit} below establishes a \emph{well-defined}
notion of limit for a point $\x_e \in W(\pi_i,\cV_e)\cup
W(\pi_i,\cV_{-e})$ when $e$ tends to $0$.
\begin{lemma}
\label{lemma:elimit}
Let $e_0$ and $\cL_i$ be as above. For $e \in ]0,e_0[$ and $\x_e \in
(W(\pi_i,\cV_e)\cup W(\pi_i,\cV_{-e}))\cap\RR^{n+1}$, there exists a
(unique) connected component $\cC$ of $\cL_i$ containing $\x_e$. If
$\cC$ is bounded, let $\x$ be the only point in $\overline{\cC}$
satisfying $f(\x) = 0$, then $\x \in \lim_{\eps}
(W(\pi_i,\cV_{\eps})\cup W(\pi_i,\cV_{-\eps}))\cap \Reps^{n+1}$. Thus,
we set $\lim_0 \x_e=\x$.

Moreover, the extension $\ext(\cC,\Reps)$ contains exactly one point
$\x_{\eps}$ such that $f(\x_{\eps})^2 = \eps^2$ and $\lim_{\eps}
\x_{\eps} = \x$.
\end{lemma}
\begin{proof}
Since $\x_e \in (W(\pi_i,\cV_e)\cup W(\pi_i,\cV_{-e}))\cap\RR^{n+1}$
and $0 < e < e_0$, we have $\x_e \in \cL_i$, the existence of
$\cC$ follows naturally. Let $\x$ be the unique point of
$\overline{\cC}$ satisfying $f=0$. Then, the notion $\lim_0$ is
well-defined. From the proof of \cite[Theorem 12.43]{BPR}, we have
that
\[\textstyle{\lim_{\eps} (W(\pi_i,\cV_{\eps})\cup
W(\pi_i,\cV_{-\eps}))\cap \Reps^{n+1} =
\pi_{\x}\left(\overline{\cW_{(0,+\infty)}}\cap \VV(t)\right)}.\]
As $\pi_{\x}\left(\overline{\cW_{(0,+\infty)}}\cap \VV(t)\right)$ is
the set of points corresponding to $f=0$ of $\cL_i$, we deduce that
$\x \in \lim_{\eps} (W(\pi_i,\cV_{\eps})\cup
W(\pi_i,\cV_{-\eps}))\cap\Reps^{n+1}$.

Since the extension $\ext(\cC,\Reps)$ is a connected component of
$\ext(\cL_i,\Reps)$ and homeomorphic to an open interval in $\Reps$,
there exists $\x_{\eps} \in \ext(\cC,\Reps)$ such that $f(\x_{\eps})^2
= \eps^2$. Moreover, since $0 = \lim_{\eps} f(\x_{\eps})^2 =
f(\lim_{\eps} \x_{\eps})^2$ and $\x$ is the only point in
$\overline{\cC}$ satisfying $f=0$, we conclude that $\lim_{\eps}
\x_{\eps} = \x$.
\end{proof}
Now, let $\cR_e$ be a roadmap associated to the algebraic set
$\cV_e\cup \cV_{-e}$, i.e. $\cR_e$ is contained in
$(\cV_e\cup \cV_{-e})\cap\RR^{n+1}$, of at most dimension one and has
non-empty intersection with every connected component of
$(\cV_e\cup\cV_{-e})\cap\RR^{n+1}$. We also require that $\cR_e$
contains $\cup_{i=1}^n (W(\pi_i,\cV_e)\cup W(\pi_i,\cV_{-e}))\cap
\RR^{n+1}$. The proposition below is the key to describe {\sf
IsIsolated}.
\begin{proposition}
\label{proposition:removeeps}
Given $e \in ]0,e_0[$ and $\cI = ]-e_0,0[\cup]0,e_0[$ as above. Let
$\cL = \cup_{i=1}^n \cL_i$ and $\x \in \cP$. Then $\x$ is not isolated
in $\cV\cap\RR^{n+1}$ if and only if there exists $\x'\in \cP$ such
that $\x$ and $\x'$ are connected in $\cP \cup \cL \cup \cR_e$.
\end{proposition}
\begin{proof}
Assume first that $\x$ is not isolated. By Proposition
\ref{proposition:main}, there exists $1\leq i\leq n$ and a connected
component $\cC_\eps$ of $(\cV_\eps\cup
\cV_{-\eps})\cap\Reps^{n+1}$, which is bounded over $\RR$, such that
$\cC_\eps\cap (W(\pi_i,\cV_\eps)\cup W(\pi_i,\cV_{-\eps}))$ contains
$\x_\eps$ and $\x'_\eps$ satisfying $\x = \lim_\eps \x_\eps \ne
\lim_\eps \x'_\eps$.  By the choice of $e_0$, there exist a
diffeomorphism $\theta: \cV_{t,\cI} \to \cV_e\times
\cI$ such that $\theta(\cW_{i,\cI}) =
\theta(\cW_{i,e}) \times \cI$.  Using \cite[Exercise
3.2]{BPR}, $\ext(\theta,\Reps)$ is a diffeomorphism between:
\begin{align*} \ext(\cV_{t,\cI},\Reps) & \cong
\ext(\cV_e,\Reps)\times \ext(\cI,\Reps),\\
\ext(\cW_{i,\cI},\Reps) & \cong \ext(\cW_{i,e},\Reps)\times
\ext(\cI,\Reps).
\end{align*}
As $\pi_\x$ is an isomorphism from $\cV_t$ to $\RR^{n+1}$,
there exists a (unique) bounded connected component $\cC_e$ of
$\cV_e\cap \RR^{n+1}$ s.t. $\cC_\eps$ is diffeomorphic to
$\ext(\cC_e,\Reps)$. Moreover, let $L$ and $L'$ be the connected
components of $\ext(\cL_i,\Reps)$ containing $\x_\eps$ and $\x'_\eps$
respectively and $\x_e$ and $\x'_e$ ($\in \ext(\cC_e,\Reps))$ be the
intersections of $\ext(\cC_e,\Reps)$ with $L$ and $L'$
respectively. Then, $\lim_\eps \x_\eps$ ($\lim_\eps L'$) connects
$\lim_\eps \x_e$ ($\lim_\eps \x'_e$) to $\x$ ($\x'$). As $\lim_\eps
\x_e$ and $\lim_\eps \x'_e$ are connected in $\cC_e$, we conclude that
$\x$ and $\x'$ are also connected in $\cP\cup\cL\cup\cR_e$.
\noindent
The reverse implication is immediate using the above techniques
\end{proof}
From Lemma \ref{lemma:elimit} and Proposition
\ref{proposition:removeeps}, any $e$ lying in the interval $]0,e_0[$
defined above can be used to replace the infinitesimal $\eps$. So, we
simply take $e = e_0/2$. For $1\leq i \leq n$, we use a subroutine
{\sf ZeroDimSolve} which takes as input
$\left\{f-e_0/2,g,\frac{\partial f}{\partial x_j}\text{ for all } j
\ne i\right\}$ to compute a zero-dimensional parametrization $\kQ_i$
such that $W(\pi_i,\cV_e) =
\{\x \in \ZZ(\kQ_i) | x_{n+1}\ne 0\}$.

To use Proposition \ref{proposition:removeeps}, we need to compute
$\cR_{e_0/2}$, which we refer to the algorithm {\sf Roadmap} in
\cite{SaSc17}. This algorithm allows us to compute roadmaps for smooth
and bounded real algebraic sets, which is indeed the case of
$(\cV_{e_0/2}\cup\cV_{-e_0/2}) \cap \RR^{n+1}$. First, we call
(another) {\sf Union} that, on the zero-dimensional parametrizations
$\kQ_i$, it computes a zero-dimensional parametrization $\kQ$ encoding
$\cup_{i=1}^n \ZZ(\kQ_i)$. Given the polynomials $f$, $g$, the value
$e_0/2$ and the parametrization $\kQ$, a combination of {\sf Union}
and {\sf Roadmap} returns a one-dimensional parametrization $\kR$
representing $\cR_{e_0/2}$.

Deciding connectivity over $\cP \cup \cL \cup \cR_e$ is done as
follows. We use {\sf Union} to compute a rational parametrization
$\kS$ encoding $\cK\cup\cR_e$. Then, with input $\mathfrak{S}$, $\kP$,
$x_{n+1}\ne 0$ and the inequalities $0 < f < e_0$, we use Newton
Puiseux expansions and cylindrical algebraic decomposition (see
\cite{Duval89,ScSh83}) following \cite{SaSc11}, taking advantage of
the fact that polynomials involved in rational parametrizations of
algebraic curves are bivariate. We denote by {\sf ConnectivityQuery}
the subroutine that takes those inputs and returns $\kP$ and isolating
boxes of the points defined by $\kP$ which are not connected to other
points of $\kP$.
\begin{algorithm}
\KwData{The polynomials $f^\bmA \in \QQ[x_1,\ldots,x_n]$ and $g\in
\QQ[x_1,\ldots,x_{n+1}]$ and the zero-dimensional parametrization
$\kP$.}
\KwResult{$\kP$ with isolating boxes of the isolated points of
$\cV^\bmA\cap \RR^{n+1}$}
\SetAlgoNoLine
$e_0\gets \text{\sf SpecializationValue}(f^\bmA,g)$ \\
\For{$1 \leq i \leq n$}{
$\kQ_i \gets \text{\sf
ZeroDimSolve}\left(\left\{f^\bmA-e_0/2,g,\diff{f^\bmA}{x_j}\text{ for
all }j\ne i\right\}\right)$ \\
$\kK_i \gets \text{\sf ParametricCurve}(f^\bmA,i)$ \\
}
$\kK \gets \text{\sf Union}(\kK_1,\ldots,\kK_n)$ \\
$\kQ \gets \text{\sf Union}(\kQ_1,\ldots,\kQ_n)$\\
$\kR \gets \text{\sf Union}(\text{\sf
  RoadMap}(f^\bmA-e_0/2,g,\kQ),\text{\sf RoadMap}(f^\bmA+e_0/2,g,\kQ))$ \\
$\kS \gets \text{\sf Union}(\kK,\kR)$ \\
$B \gets \text{\sf ConnectivityQuery}(\mathfrak{S},\kP,x_{n+1}\ne 0,0<
f^\bmA < e_0)$\\
\Return $(\kP,B)$
\caption{{\sf IsIsolated}}
\label{algo:isisolated}
\end{algorithm}

\section{Complexity analysis}
\label{section:complexity}
All complexity results are given in the number of arithmetic
operations in $\QQ$. Hereafter, we assume that a generic enough matrix
$\bmA$ is found from a random choice. In order to end the proof of
Theorem~\ref{thm:main}, we now estimate the arithmetic runtime
of the calls to {\sf Candidates} and {\sf IsIsolated}.
\vspace{-0.1cm}
\paragraph{Complexity of Algorithm~\ref{algo:candidates}}
Since $W(\pi_i,\cH_{\eps}^\bmA)$ is the finite algebraic set
associated to $\left\langle f^\bmA-e, \diff{f^\bmA}{x_j}\text{ for all
} j\ne i\right\rangle$, its degree is bounded by $d(d-1)^n$
\cite{Heintz83}.  Consequently, the degree of the output
zero-dimensional parametrization lies in $d^{O(n)}$. Using
\cite[Theorem 6]{Sa05} (which is based on the geometric resolution
algorithm in \cite{GLS01}), it is computed within $d^{O(n)}$
arithmetic operations in $\QQ$. The last step which takes
intersections is done using the algorithm in \cite[Appendix
J.1]{SaSc17} ; it does not change the asymptotic complexity.

We have seen that {\sf GetNormBound} reduces to isolate the real roots
of a zero-dimensional parametrization of degree $d^{O(n)}$. This can
be done within $d^{O(n)}$ operations by Uspensky's algorithm
\cite{RoZi03}.
\paragraph{Complexity of Algorithm~\ref{algo:isisolated}}
Each call to {\sf SpecializationValue} reduces to computing critical
values of $\pi_i$ of a smooth algebraic set defined by polynomials of
degree $\leq d$.  This is done using $(nd)^{O(n)}$ arithmetic
operations in $\QQ$ (see \cite{GS14}). Using \cite{GLS01} for {\sf
ZeroDimSolve} and \cite{Sch03} for {\sf ParametricCurve} does not
increase the overall complexity. The loop is performed $n$ times ;
hence the complexity 
lies in $(nd)^{O(n)}$. All output zero-dimensional parametrizations
have degree bounded by $d^{O(n)}$. Running {\sf Union} on these
parametrizations does not increase the asymptotic complexity. One gets
then parametrizations of degree bounded by $nd^{O(n)}$.  Finally,
using \cite{SaSc17} for {\sf Roadmap} uses $(nd)^{O(n\log(n))}$
arithmetic operations in $\QQ$ and outputs a rational parametrization
of degree lying in $(nd)^{O(n\log(n))}$. The call to {\sf
ConnectivityQuery}, done as explained in \cite{SaSc11} is polynomial
in the degree of the roadmap.

The final steps which consist in calling {\sf Removes} and
undoing the change of variables 
does not change the asymptotic complexity.

Summing up altogether the above complexity estimates, one obtains an
algorithm using $(nd)^{O(n\log(n))}$ arithmetics operations in $\QQ$
at most. This ends the proof of Theorem~\ref{thm:main}.
\vspace{-0.2cm}
\section{Experimental results}\label{sec:experiments}
We report on practical performances of our algorithm. Computations
were done on an Intel(R) Xeon(R) CPU E3-1505M v6 @ 3.00GHz with 32GB
of RAM. We take sums of squares of $n$ random dense quadrics in $n$
variables (with a non-empty intersection over $\RR$) ; we obtain {\em
dense quartics} defining a finite set of points. Timings are given
in seconds (s.), minutes (m.), hours (h.) and days (d.).

We used Faug\`ere's {\sc FGb} library for computing Gr\"obner bases in
order to perform algebraic elimination in Algorithms~\ref{algo:main},
\ref{algo:candidates} and \ref{algo:isisolated}.  We also used our C
implementation for bivariate polynomial system solving (based on
resultant computations) which we need to analyze connectivity queries
in roadmaps. Timings for Algorithm~\ref{algo:candidates} are given in
the column {\sc cand} below. Timings for the computation of the
roadmaps are given in the column {\sc rmp} and timings for the
analysis of connectivity queries are given in the column {\sc qri}
below.

Roadmaps are obtained as the union of critical
loci of some maps in slices of the input variety \cite{SaSc17}. We
report on the highest degree of these critical loci in the column {\sc
srmp}. The column {\sc sqri} reports on the maximum degree of the
bivariate zero-dimensional system we need to study to analyze
connectivity queries on the roadmap.

None of the examples we considered could be tackled using the
implementations of Cylindrical Algebraic Decomposition algorithms in
Maple and Mathematica.

We also implemented \cite[Alg. 12.16]{BPR} using the {\sc Flint} C
library with evaluation/interpolation techniques instead to tackle
coefficients involving infinitesimals. This algorithm only
computes sample points per connected components. {\em That
implementation was not able to compute sample points of the
input quartics for any of our examples}. We then report in the column
[BPR] on the degree of the zero-dimensional system which is expected
to be solved by \cite[BPR]{BPR}.  This is to be compared with columns
{\sc srmp} and {\sc sqri}.
\begin{center}
\begin{tabular}{cccccccc}
$n$& {\sc cand} & {\sc rmp} & {\sc qri} & total & {\sc srmp} & {\sc
                                                               sqri} & [BPR] \\
  \hline
  4& $2$ s.& $15$ s.& $33$ s.&50 s. & 36 & 359 & 7290 \\
  5& $<10$ min. & 1h. & 7h. & 8 h. & 108 & 4644 & 65\ 610\\
  6 & $<12h$& 2 d. &18 d. & 20 d. & 308&47952 & 590\ 490 \\
\end{tabular}
\end{center}
\bibliography{mybib}
\bibliographystyle{acm}

\end{document}